\documentclass[11pt]{article}

\usepackage{amsthm,amsmath,amssymb,latexsym}
\usepackage{euscript}
\usepackage{array}
\usepackage{times}
\usepackage{graphicx,color,epsfig}

\usepackage[]{epsf}
\usepackage{amsfonts}
\usepackage{xspace}
\usepackage{cite}

\def\R{{\mathbb R}}

\newtheorem{theorem}{Theorem}[section]

\newtheorem{definition}[theorem]{Definition}

\newcommand\set[1]{\left\{ #1 \right\}}
\newcommand\sett[2]{\left\{ \left. #1 \;\right\vert #2 \right\}}

\newcommand\floor[1]{\left\lfloor #1 \right\rfloor}

\newcommand\remove[1]{}

\newcommand\erdos{{Erd{\H o}s}\xspace}
\newcommand\szemeredi{{Szemer\'edi}\xspace}
\newcommand\Cst{C_{\mathrm{SzTr}}}
\newcommand\cst{c_{\mathrm{SzTr}}}
\newcommand\Elekes{\mathrm{Elekes}}
\newcommand\Erdos{\mathrm{Erdos}}

\begin{document}

\title{The Constant of Proportionality in Lower Bound Constructions of
Point-Line Incidences}
\author{Roel Apfelbaum}

\maketitle

\begin{abstract}
Let $I(n,l)$ denote the maximum possible number of incidences between $n$ points
and $l$ lines.
It is well known that $I(n,l) = \Theta(n^{2/3}l^{2/3} + n + l)$
\cite{ST,Er-points-lines,El-sums}.
Let $\cst$ denote the lower bound on the constant of proportionality of the
$n^{2/3}l^{2/3}$ term.
The known lower bound, due to Elekes \cite{El-sums}, is
$\cst \ge 2^{-2/3} = 0.63$.
With a slight modification of Elekes' construction, we show that it can give a
better lower bound of $\cst \ge 1$, i.e., $I(n,l) \ge n^{2/3}l^{2/3}$.
Furthermore, we analyze a different construction given by \erdos
\cite{Er-points-lines}, and show its constant of proportionality to be even
better, $\cst \ge 3/(2^{1/3}\pi^{2/3}) \approx 1.11$.
\end{abstract}

\section{Overview}
Let $P$ be a set of $n$ points in $\R^2$, and let $L$ be a family of $l$ lines
in $\R^2$.
We denote the number of incidences between these points and lines by $I(P,L)$.
We denote by $I(n,l)$ the maximum of $I(P,L)$ over all sets $P$ of $n$ points,
and families $L$ of $l$ lines.
The \szemeredi-Trotter bound \cite{ST} asserts that
$I(n,l) = O(n^{2/3}l^{2/3} + n + l)$
(See also \cite{CEGSW, Sz} for simpler proofs).
For values of $n$ and $l$ such that $\sqrt{n} \le l \le n^2$, the
$n^{2/3}l^{2/3}$ term dominates, so the bound becomes
$I(n,l) = O(n^{2/3}l^{2/3})$.
In more detail, we have:
\begin{theorem}[\szemeredi and Trotter \cite{ST}] \label{thm:st}
There exists a constant $\Cst$ such that, for any set $P$ of $n$ points, and any
family $L$ of $l$ lines, if $\sqrt{n} \le l \le n^2$, then the number of
incidences between the points and lines is at most
\[
I(P,L) \le \Cst n^{2/3} l^{2/3}.
\]
\end{theorem}
\noindent
The known upper bound on $\Cst$ at present, due to Pach et al. \cite{PRTT}, is
$\Cst \le 2.5$.
The bound of Theorem \ref{thm:st} is asymptotically tight, as shown in different
lower bound constructions by \erdos \cite{Er-points-lines} and Elekes
\cite{El-sums}.
We state this claim more formaly as follows.
\begin{theorem}[\erdos \cite{Er-points-lines}, Elekes \cite{El-sums}]
\label{thm:low}
There exists a constant $\cst > 0$, such that, for infinitely many values of $n$
and $l$, where $\sqrt{n} \le l \le n^2$, there exist pairs $(P,L)$, where
$P$ is a set of $n$ points, and $L$ is a family of $l$ lines, such that the
number of incidences between the points and lines is at least
\[
I(P,L) \ge \cst n^{2/3} l^{2/3}.
\]
\end{theorem}
\noindent
The known lower bound on $\cst$, due to Elekes \cite{El-sums}, is
$\cst \ge 2^{-2/3} = 0.63$.

In this paper we improve the estimate of $\cst$.
We modify Elekes' construction, and show that this modification gives a lower
of $\cst \ge 1$.
Next, we analyze the construction of \erdos \cite{Er-points-lines}, and show its
constant of proportionality to be even better,
$\cst \ge 3/(2^{1/3}\pi^{2/3}) \approx 1.11$.
This is an improvement upon a previous analysis of the \erdos construction
\cite{PT}, which gives the bound $\cst \ge (3/(4\pi^2))^{1/3} \approx 0.42$.

\section{The Elekes construction}
\begin{figure}
\begin{center}
\includegraphics[width=400pt,height=270pt]{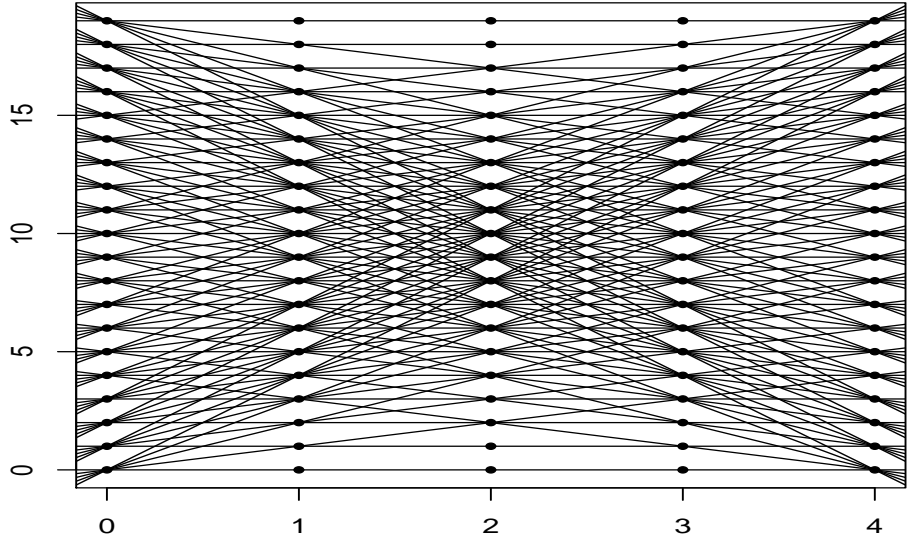}
\end{center}
\caption{\sf An $\Elekes(5,4)$ configuration. $n = 100$ points, $l = 100$ lines,
and $I = 500$ incidences.} \label{fig:elekes}
\end{figure}
Elekes \cite{El-sums} gave the following lower bound construction.
Let $k$ and $m$ be some positive integers.
Put $P = \set{1, \ldots, k} \times \set{1, \ldots, 2km}$, and put $L$ to be all
lines $y = ax + b$, where $a \in \set{1, \ldots, m}$, and
$b \in \set{1, \ldots, km}$.
There are $n = |P| = 2k^2m$ points and $l = |L| = km^2$ lines here, and each
line is incident to exactly $k$ points, so $I = I(P,L) = k^2m^2$.
It is then easy to verify that $I = 2^{-2/3}n^{2/3}l^{2/3}$, and also, whenever
$m > 1$, that $\sqrt{n} \le l \le n^2$.
This gives a lower bound on the $\cst$ constant from Theorem \ref{thm:low} of
$\cst \ge 2^{-2/3} \approx 0.63$.

We present a slightly different construction from the above.
It is similar in principle, but more exhaustive.
\begin{definition}
Let $k$ and $m$ be some positive integers.
We denote by
$$\Elekes(k,m) = (P,L)$$
the following set of points $P$, and family of lines $L$.
$P$ is defined as a $k \times km$ lattice section:
\[
P = \set{0, \ldots, k-1} \times \set{0, \ldots, km-1},
\]
and $L$ is defined as all $x$-monotone lines that contain $k$ points of $P$.
\end{definition}
With this definition of $\Elekes(k,m)$, we have $I(P,L) \ge |P|^{2/3}|L|^{2/3}$,
and hence, $\cst \ge 1$. More formally:
\begin{theorem}
Let $P$ and $L$ respectively be the points and lines of an $\Elekes(k,m)$
configuration, for some positive integers $k>1$ and $m$.
Let us denote the number of points by $|P| = n$, the number of lines by
$|L| = l$, and the number of incidences between them by $I(P,L) = I$.
Then $I \ge n^{2/3}l^{2/3}$.
\end{theorem}
\begin{proof}
The lines of $L$ have the form $y = ax + b$ with integer parameters as follows.
The $b$ parameter is an integer in the range
\[
0 \le b \le km-1,
\]
and the $a$ parameter, given $b$, is restricted as follows.
For $x=k-1$ we have $0 \le a(k-1) + b \le km-1$, or
\[
-\frac{b}{k-1} \le a \le m + \frac{m-1}{k-1} - \frac{b}{k-1}.
\]
The difference between the upper and lower bounds of $a$ is $m + (m-1)/(k-1)$,
and the number of integer values in this range is either
$m + \floor{(m-1)/(k-1)}$, or $m + 1 + \floor{(m-1)/(k-1)}$.
The latter case happens about $1 + ((m-1) \mod (k-1))$ out of $k-1$ times.
The number of lines, resulting from multimplying the number of $b$-values by the
number of $a$-values, is
\[
l \approx km \left(m + \floor{\frac{m-1}{k-1}}
				+ \frac{1 + ((m-1) \mod (k-1))}{k-1}\right),
\]
and in any event it is greater than $km^2$,
\[
l \ge km^2.
\]
The number of points is
\[
n = k^2m.
\]
It then follows that
\[
k \ge \frac{n^{2/3}}{l^{1/3}}.
\]
Since each line is incident to $k$ points, the number of incidences comes out
\[
I = lk \ge n^{2/3}l^{2/3}
\]
as claimed. This completes the proof.
\end{proof}
From this theorem it follows that $\cst \ge 1$.
Note that an $\Elekes(k,k-1)$ has an equal number of points and lines,
$n = l = k^2(k-1)$, and $I = k^3(k-1) \approx n^{4/3}$ incidences.

\section{The \erdos construction}
\erdos \cite{Er-points-lines} considered $n$ points on a
$n^{1/2} \times n^{1/2}$ lattice section, together with the $n$ lines that
contain the most points.
He noted that there are $\Theta(n^{4/3})$ incidences in this configuration, and
conjectured that it is asymptotically optimal.
His conjecture was settled in the affirmative as a corollary of the
\szemeredi-Trotter bound \cite{ST}.
Pach and T\'oth \cite{PT} analyzed, in more generality, the square lattice
section together with the lines with the most incidences, where the number of
lines $l$ is not necessarily equal to the number of points $n$.
Their analysis yielded the bound $I \ge 0.42n^{2/3}l^{2/3}$.
In this section we will analyze the same setting in a different way and get an
improved bound of $I \ge 1.11n^{2/3}l^{2/3}$, i.e., $\cst \ge 1.11$.

First, we give a formal definition of the \erdos construction.

\begin{definition}
For two positive integers $k$ and $m$, we denote by $$\Erdos(k,m) = (P,L)$$
the following set of points $P$, and family of lines $L$.
We put $P$ to be a $k \times k$ lattice section:
\[
P = \set{0, \ldots, k-1}^2.
\]
Next, we put $L$ to be all lines of the form $ax + by = c$ that pass through the
bounding square of $P$, where:
\begin{enumerate}
\item $a$, $b$, and $c$ are integers.
\item $a$ and $b$ are coprime.
\item $a \ge 0$.
\item $|a|+|b| \le m$.
\end{enumerate}
\end{definition}
Under this definition, $L$ is not quite the family of lines with the most
incidences with respect to $P$, but rather, an approximation of it.
Indeed, there are lines here, such as $x + y = 0$, with  just one incidence.
There are even lines with no incidences, like $2x + 3y = 1$ (this line exists
whenever $k \ge 2$, and $m \ge 5$).
However, most lines do have many incidences,
which gives us the following result.
\begin{figure}
\begin{center}
\includegraphics[width=360pt,height=400pt]{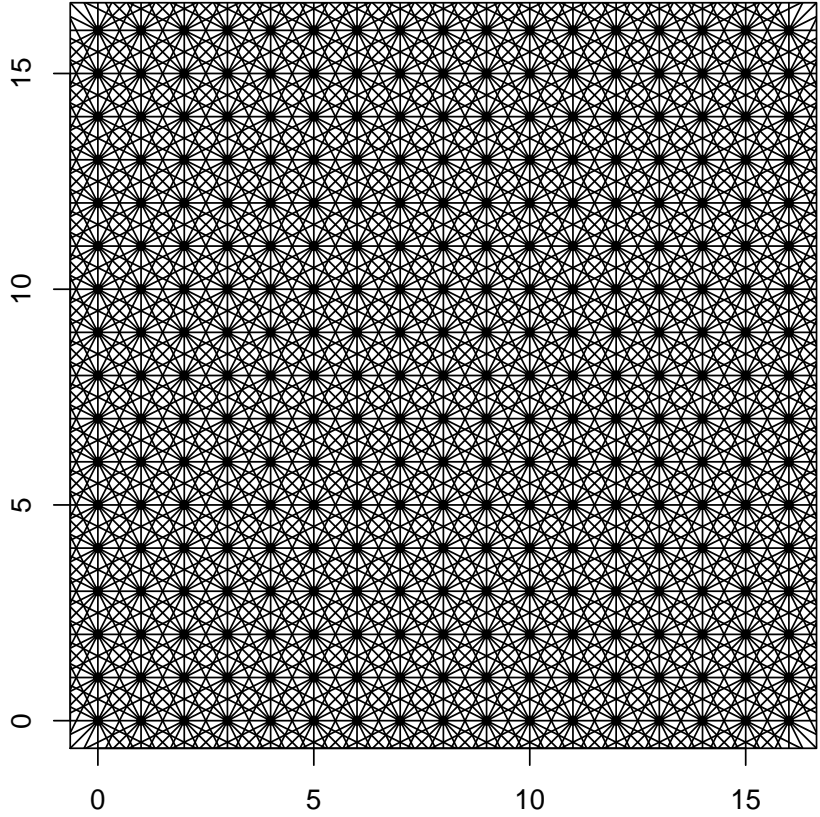}
\end{center}
\caption{\sf An $\Erdos(17,3)$ configuration. $n = 289$ points, $l = 296$ lines,
and $I = 2312$ incidences.} \label{fig:erdos}
\end{figure}
\begin{theorem} \label{thm:erd}
Let $P$ and $L$ respectively be the points and lines of an $\Erdos(k,m)$
configuration, for some positive integers $k$ and $m$.
Let us denote the number of points by $|P| = n$, the number of lines by
$|L| = l$, and the number of incidences between them by $I(P,L) = I$.
Then $I \approx \frac{3}{2^{1/3}\pi^{2/3}} n^{2/3}l^{2/3}$.
\end{theorem}
The notation $\Phi \approx \Psi$, where both expressions depend on some
set of variable $x_1, x_2, \ldots$ is shorthand for
$\lim_{x_1 \to \infty, x_2 \to \infty, \ldots} (\Phi / \Psi) = 1$.
That is, as the independent variables (in the case of Theorem \ref{thm:erd}, $k$
and $m$) grow larger and larger, the ratio between the two expressions ($I$ and
$\frac{3}{2^{1/3}\pi^{2/3}} n^{2/3}l^{2/3}$, in the case of
Theorem \ref{thm:erd}) gets closer and closer to one.
\begin{proof}
The number of points is $n = k^2$.
The probability of a random pair $(a,b)$ to be coprime is about
$\frac{6}{\pi^2}$ \cite{Coprime}.
There are $(m+1)^2$ integer pairs in the range
$\sett{(a,b)}{|a|+|b| \le m, a \ge 0}$,
so there are about $\frac{6m^2}{\pi^2}$ coprime pairs.
Each pair $(a,b)$ determines the direction of a pencil of parallel lines,
$ax + by = c$, and each of the $k^2$ points is incident to a line in each of
these directions.
That is, each point is incident to about $\frac{6m^2}{\pi^2}$ lines, so in total
\[
I \approx \frac{6k^2m^2}{\pi^2}.
\]
It remains to estimate the number of lines.
Consider a positive coprime pair $(a,b)$. This pair generates lines
$ax + by = c$, where:
\begin{enumerate}
\item The minimal value of $c$ is $0$, and the line $ax + by = 0$ passes through
      $(0,0) \in P$.
\item The maximal value of $c$ is $(a+b)(k-1)$, and the line
      $ax + by = (a+b)(k-1)$ passes through $(k-1,k-1) \in P$.
\end{enumerate}
It follows that there are $(|a|+|b|)(k-1)+1$ values of $c$ that generate lines
that pass through the square.
This number of lines is true also for negative $b$ with a different range of
$c$-values.
The total number of lines $|L| = l$ is thus
\begin{eqnarray}
l & = & \sum_{a,b} ((|a|+|b|)(k-1)+1) \label{eq:l1} \\
	& \approx & \sum_{j=1}^m\sum_{|a|+|b|=j} j(k-1) + \frac{6m^2}{\pi^2} \label{eq:l2} \\
	& \approx & \sum_{j=1}^m \frac{12j}{\pi^2}j(k-1) + \frac{6m^2}{\pi^2} \label{eq:l3} \\
	& \approx & \frac{12(k-1)}{\pi^2}\sum_{j=1}^m j^2 + \frac{6m^2}{\pi^2} \label{eq:l4} \\
	& \approx & \frac{4m^3(k-1)}{\pi^2} + \frac{6m^2}{\pi^2}. \label{eq:l5}
\end{eqnarray}
(\ref{eq:l1}) is a sum over all coprime pairs $(a,b)$ as above.
(\ref{eq:l2}) is the same sum in a different order of summation.
In (\ref{eq:l3}) we estimate the number of coprime pairs $(a,b)$ such that
$|a|+|b|=j$ as follows. There are $2j+1$ integer pairs $(a,b)$, such that
$a \ge 0$ and $|a|+|b|=j$, and the probability of a pair from this subset
to be coprime is, as already noted, $6/\pi^2$, so there should be an
expected number of $(12j + 6)/\pi^2 \approx 12j/\pi^2$ coprime pairs.
In (\ref{eq:l5}) we use the approximation $\sum_{j=1}^m j^2 = m(m+1)(2m+1)/6
\approx m^3/3$.
The dominant term in the final equation is
\[
l \approx \frac{4m^3k}{\pi^2}.
\]
From the values of $n, l$, and $I$ in terms of $k$ and $m$, we get that
\[
I \approx \frac{3}{2^{1/3}\pi^{2/3}}n^{2/3}l^{2/3}
\]
as claimed. This copmletes the proof.
\end{proof}
From Theorem \ref{thm:erd} it follows that
$\cst \ge \frac{3}{2^{1/3}\pi^{2/3}} \approx 1.11$.

\bibliographystyle{plain}

\end{document}